\documentclass[sigconf]{acmart}

\setcopyright{none}
\acmConference[]{}{}{} 

\usepackage{graphicx}
\usepackage{balance}  

\usepackage[ruled,vlined,linesnumbered]{algorithm2e}
\usepackage{tikz}
\usetikzlibrary{arrows,arrows.meta,automata,snakes,positioning,shapes,shapes.geometric,shadows,calc}
\tikzset{>={Latex[width=3mm,length=3mm]}}

\usepackage{subfigure}

\usepackage{amsmath,amssymb}
\usepackage{stmaryrd}
\usepackage{xspace,latexsym,epic,eepic,graphicx}
\usepackage{scalefnt}
\usepackage{stmaryrd}
\usepackage{xspace}
\usepackage{relsize}
\usepackage{url}
\usepackage{multirow}
\usepackage{enumitem}

\newcommand{\Frontier}{\ensuremath{\mathit{Open}}\xspace}
\newcommand{\fs}[1]{\fontsize{#1}{#1}\selectfont}
\newcommand{\angled}[1]{\ensuremath{\langle #1\rangle}}
\newcommand{\Succ}{\ensuremath{\mathit{Succ}}}
\newcommand{\sstart}{\ensuremath{s_{\mathit{start}}}\xspace}

\renewcommand{\epsilon}{\varepsilon}

\newcommand{\sem}[1]{\llbracket{} #1 \rrbracket{}}

\renewcommand{\phi}{\varphi}

\newcommand{\I}{\mathcal{I}}

\newcommand{\Lit}{\mathcal{L}}

\newcommand{\G}{\mathcal{G}}

\newcommand{\adoc}{\textit{adoc}}

\newcommand{\init}{\textit{init}}
\newcommand{\Seen}{\textit{Seen}}
\newcommand{\Gtemp}{G_\textit{temp}}

\newcommand{\Gastar}{\ensuremath{G_{\text{A*}}}\xspace}

\newcommand{\ignore}[1]{}

\newcommand{\MS}{\texttt{M.Stonebraker}}

\newcommand{\qA}{\textbf{q\_Coauthor}\xspace} 
\newcommand{\qB}{\textbf{q\_Bacon}\xspace}    
\newcommand{\qC}{\textbf{q\_NATO}\xspace}     

\title{Evaluating navigational RDF queries over the Web}

\author{Jorge Baier, Dietrich Daroch, Juan L. Reutter, Domagoj Vrgo\v{c}}

\affiliation{Pontifica Universidad Cat\'olica de Chile and Center for Semantic Web Research}

\begin{abstract}
Semantic Web, and its underlying data format RDF, lend themselves naturally to navigational querying due to their graph-like structure. This is particularly evident when considering RDF data on the Web, where various separately published datasets reference each other and form a giant graph known as the Web of Linked Data. And while navigational queries over singular RDF datasets are supported through SPARQL property paths, not much is known about evaluating them over Linked Data. In this paper we propose a method for evaluating property path queries over the Web based on the classical AI search algorithm A*, show its optimality in the open world setting of the Web, and test it using real world queries which access a variety of RDF datasets available online and that are not necessarily known in advance.
\end{abstract}

\begin{document}

\maketitle

\sloppy

\section{Introduction}
\label{sec:Introduction}

The Resource Description Framework (RDF) \cite{manola2004rdf} is the World Wide Web consortium (W3C) standard for representing Semantic Web data. 
In essence, an RDF graph is a set of triples of internationalised resource identifiers (IRIs), where the first and last of them represent entity resources, and the middle one relates these resources, just as is it done in graph databases \cite{AG08}. The official query language for RDF databases is SPARQL \cite{sparql11}. 

To answer the need for including navigational features into SPARQL, the latest version of the language includes \emph{property paths}, 
a set of queries that can be seen as the analogues of established graph database languages such as regular path queries and two-way regular path queries \cite{CGLV00}. Consequently, property paths are already supported by the vast majority of existing SPARQL engines (e.g., \cite{jena,virtuoso,blazegraph}). 
The inclusion of navigational queries is also present in most other graph database models (see e.g. \cite{B13,AG08}).


Besides the traditional approach where one issues a query over a (set of) graph databases, 
the community has further raised the need for a fundamentally different way of querying 
RDF data: 
to obtain answers of queries over the whole corpus of RDF data present on the Web and linked together into what is known as the {\em Web of Linked Data}, in a distributed 
way and without assuming any mediation nor centralised organisation in control of the data, following the \emph{Linked Data Principles} \cite{BLBH09}. 


The fundamental property of RDF data that makes this querying possible is that the IRIs in RDF documents published online should be \emph{dereferenceable}.
This basically means that by accessing any given IRI, we obtain a new RDF document describing its neighbourhood (or a part of it) in the Linked Data graph.
Let us explain how this works using the online RDF documents published by DBLP, one of the simplest datasets now forming part of the Web of Linked Data. 
In the RDF representation of DBLP, each researcher is given a unique IRI, as
well as each paper. The authorship relation indicating that an author $A$ wrote
a paper $P$ is then represented by the triple $\{P,\texttt{dc:creator},A\}$.
The IRI for each author, then serves as a good starting point for investigating the DBLP dataset, as dereferencing their IRI will intuitively give us all the 
papers written by this author. 
For example, if we dereference the IRI \MS
, representing Michael Stonebraker, we obtain a document containing, amongst other things, the following triples 

\smallskip
{\centering \small
\begin{tabular}{ccc}
\MS & \texttt{foaf:name} & ``M. Stonebraker"  \\
\texttt{inTods:StonebrakerWKH76} & \texttt{dc:creator} & \MS  \\
\texttt{inSigmod:PavloPRADMS09} & \texttt{dc:creator} & \MS  
\end{tabular}}
\smallskip

These triples indicate that  \MS\ is the author of the papers represented by IRIs \texttt{inTods:StonebrakerWKH76} and \texttt{inSigmod:PavloPRADMS09}, 
and that the
name of the entity represented by \MS\  is indeed ``Michael Stonebraker". 
Suppose now that we need to retrieve the names of all the co-authors of Michael Stonebraker. It is very easy to do this using the linked data infrastructure: We first dereference the IRI \MS, obtaining an RDF document that contains, in particular, a triple $\{P,\texttt{dc:creator},\MS\}$ for each paper $P$ authored by M. Stonebraker. Then we just need to dereference each of the IRIs of these papers: dereferencing each 
of these IRIs $P$ gives us triples of the form  $\{P,\texttt{dc:creator},A\}$, and now we know that $A$ is a coauthor of M. Stonebraker. The last step is to further dereference the IRI of each of these researchers, to look for a triple $\{A,\texttt{foaf:name},N\}$ that indicates the name of the researcher (in this case $N$). 

Of course, the query looking for co-authors of Michael Stonebraker can be seen as a fixed pattern: namely, it is a path of length two, starting in the IRI \MS\ and traversing the edge \texttt{dc:creator} backwards (thus reaching a paper written by Michael Stonebraker), and then traversing the \texttt{dc:creator} edge forwards to reach one of his co-authors. But what happens when we want to generalise this query and obtain the collaboration reach of Michael Stonebraker, that is, his co-authors, the co-authors of his co-authors, their co-authors, etc? This is similar to the popular notion of Erd\H{o}s number, but this time starting with a different author. To answer such a query a fixed length path will no longer suffice, since we do not know the distance between the starting node and the ending node in advance. We therefore need to use property paths; in this case  
this would be done using the query

\medskip
{\small
\begin{tabular}{ccc}
\MS & \texttt{(\textasciicircum dc:creator/dc:creator)*} & \texttt{?x},
\end{tabular}}
\medskip

\noindent which repeats the simple path from one author to a paper (using \texttt{\textasciicircum dc:creator} to follow an edge labelled \texttt{dc:creator} in a reverse direction) and then to another author (using \texttt{dc:creator}) an arbitrary number of times, as signified by the star operator \texttt{*}. The idea is as before, but now once a co-author is retrieved, search does not stop, but continues with this (co-)author as the starting node. 

When evaluating these queries we have only dereferenced and fetched the documents that we needed in order to answer the query, and thus we are taking full advantage of the nature of Linked Data. There is another fundamental advantage of this approach: we can cross between different domains without any effort by using the infrastructure of the Web, which happens when a dereferenced IRI links to another IRI residing on a different server. 
This is in contrast with, for instance, issuing a single distributed query to a centralised endpoint, since we can access an arbitrary number of different sources.
Furthermore, we can access data that is not 
published on dedicated endpoints; all that we need is data published on the standard Web architecture. Up to our best knowledge, this framework of distributed, decentralised and ungoverned querying has not been considered before the advent of Linked Data. 

The advantages of these approaches have led the Semantic Web community to 
investigate the fundamentals of querying over the Web \cite{HBF09}, and developing  algorithms for answering SPARQL queries over Linked Data \cite{HartigO16,UmbrichHPD15}.  
%
Unfortunately, despite the potential that property paths could have in Web querying, most of the algorithms developed in this context focus on the pattern matching features of SPARQL, and do not consider property paths. Indeed, the majority of studies about property paths only consider how they work over a single centralised dataset \cite{GubichevBS13,YakovetsGG16,KRRV15,ACP12}. And while the need for understanding how property paths might work over the Web has repeatedly been raised by the research community \cite{HP15,HPe15,BaierDRV16}, previous studies have mostly  focused on understanding appropriate semantics and/or proposing new languages to help users navigate the Web, instead of describing the algorithms computing these answers. The only exception is \cite{FPG15}, suggesting a basic depth-first search algorithm in the context of 
NautiLOD queries: a language proposal that extends property paths. Therefore, the main objective of this paper is to answer the question: {\em How can one efficiently evaluate property path queries over the Web of Linked Data?}

%

\medskip
\noindent
\textbf{Contributions.}
%
Our main contribution is an algorithm for efficiently retrieving answers to property path queries over Linked Data. Our solution is based on the observation that evaluating property paths can be seen as a search problem over an initially unknown graph. Indeed, in the examples above we start from one known IRI (\MS) and begin exploring its neighbours guided by the query we are trying to answer. But this problem has been well studied by the Artificial Intelligence community, and it is generally agreed that the most appropriate solution here is an heuristic-search algorithm such as A* \cite{HartNR68,Pearl84}. In this paper we propose a variant of A* for the setting of Linked Data by using the property path we are trying to answer as a heuristic to guide our search. The main advantages of this approach are the following:
\begin{itemize}[leftmargin=*]
   \setlength{\parskip}{0pt}
 \setlength{\itemsep}{0pt}
 \setlength{\topsep}{0pt}
 \setlength{\partopsep}{0pt}
 \setlength{\labelwidth}{0pt}
 \addtolength{\leftmargin}{-10em}
 \setlength{\itemsep}{0pt}
\item[-] It allows to overcome shortcomings of basic graph traversal algorithms such as depth-first search (DFS) and breadth-first search (BFS). In fact, we show that A* dominates BFS and DFS, and that it is optimal with respect to the part of the graph that became available during the search. This, in some sense, is the best we can hope for in the open-world setting of the Web.
\item[-] It does not only allow to find pairs of nodes connected by a property path, but it can also return (one of the) shortest paths which witness this connection: a feature that existing SPARQL engines are currently lacking. 
\item[-] It is very robust when evaluating property paths live over the Web infrastructure, and can often answer queries which fail even on SPARQL 
implementations executed over a local dataset. 

\end{itemize}

Apart from describing the basic implementation of the A* algorithm and proving its optimality, we also develop several optimisations geared towards query answering in the Linked Data setting. Most notably, we show that dereferencing multiple IRIs in parallel can speed up the computation of property paths significantly. 
Finally, we describe how our implementation runs over the Web of Linked data using a number of real-world queries which utilise different RDF datasets. 
We compare our approach to BFS and DFS-based algorithms and their parallel versions, showing that A* is superior when it comes to querying over the Web. 

\smallskip
\noindent
\textbf{Outline.} We formalise Linked Data and property paths in Section \ref{sec:Preliminaries}. In Section \ref{sec:ComputingPathQueries} we describe how DFS and BFS can be used to answer property path queries and what are their shortcomings.  In Section \ref{sec:AISearch} we introduce the A* algorithm and show its optimality. Optimisations are presented in Section \ref{sec:Optimisations}, and real-world experiments in Section \ref{sec:Evaluation}. 
We conclude in Section \ref{sec:Conclusions}.

\section{Preliminaries}
\label{sec:Preliminaries}


\noindent{\bf RDF graphs.}
Let $\I$ and $\Lit$ be countably infinite disjoint sets of \emph{IRIs} and  \emph{literals}, respectively.
An \emph{RDF triple} is a triple $(s, p, o)$  from $(\I \cup \Lit) \times \I \times (\I \cup \Lit)$, where $s$ is called \emph{subject}, $p$ \emph{predicate}, and $o$ \emph{object}. An \emph{(RDF) graph} is a finite set of RDF triples. For simplicity we only deal with RDF documents that do not contain blank nodes.

\smallskip

\noindent{\bf Linked Data.}
We are interested in computing navigational queries over the 
wide body of RDF documents published on the Web that comprise what is known as the Web of Linked Data. 
As customary in the literature (see e.g.~\cite{AV97, H12}), we treat this corpus of documents as 
a tuple $W = (\G,\adoc)$, where $\G$ is a set of RDF graphs and $\adoc: \I \rightarrow \G \cup \{\emptyset\}$
is a function that assigns graphs in $\G$ to some IRIs, and the empty graph to the rest of the IRIs. Note that previous work (e.g. \cite{H12}) usually defines $\adoc$ as a partial function. We adopt instead the convention that $\adoc(u) = \emptyset$ whenever
$\adoc$ is not defined for $u$, as it simplifies the presentation.

The intuition behind this definition is that $\G$ represents the set of documents on the Web of Linked
data, and $\adoc$ captures dereferencing; that is, $\adoc(u)$ gives us the neighbours of $u$ in $\G$. Note that $\G$ is usually not available
and has to be retrieved by looking up IRIs with $\adoc$.
\begin{example}
We can now formalise the operations performed in the introduction over the linked data architecture of DBLP. 
Starting with the IRI {\em \MS}, we can invoke $\adoc$ on this IRI to fetch its associated graph
\begin{multline*}
{\scriptsize \adoc(\text{\em \MS}) = } \\ 
$\ \ \ \ \ \ \ \ \ \  $  \begin{cases} 
{\scriptsize
\begin{tabular}{ccc}
\em \MS & \em \texttt{foaf:name} & \em ``M. Stonebraker"  \\
\em \texttt{inSigmod:PavloPRADMS09} & \em \texttt{dc:creator} & \em \MS  \\
\vdots & \vdots & \vdots
\end{tabular}}\end{cases}
\end{multline*}
When looking for the coauthors of M. Stonebraker,  
we might want to fetch $\adoc(\text{\em \texttt{inSigmod:PavloPRADMS09}})$, which will give us a graph containing, amongst other things, 
triples of the form $(\text{\em \texttt{inSigmod:PavloPRADMS09}}, \text{\em \texttt{dc:creator}}, A)$, with $A$ being the IRI of the authors of the paper.
To get the name of $A$ we fetch the graph $\adoc(A)$ and look for the triple with {\em \texttt{foaf:name}} as the predicate.  
\end{example}





\noindent{\bf Property Paths.}
Navigational queries over graph databases commonly ask for paths that satisfy certain properties.
The most simple of them
correspond to {\em regular path queries}, or {\em RPQs\/}~\cite{ABS99,CMW87}, which select pairs of nodes  connected by a path
conforming to a regular expression, and \emph{2-way} regular path queries, or 2RPQs \cite{CGLV00}, which extend RPQ with the ability to traverse an edge backwards. 
SPARQL 
features a class of navigational queries known as \emph{property paths}, which are
themselves an extension of the well known class of 2RPQs. For readability we assume we deal only with 2RPQs, adopting the 
formalisation in~\cite{KRRV15}. Note however that our algorithms (and our implementation) work for all property path expressions. 

Formally, we define property paths by the grammar
$$
e \ := \ u \mid e^- \mid e_1\cdot e_2 \mid e_1 + e_2 \mid e^* \mid e?,
$$
where $u$ is an IRI in $\I$. 
The semantics of property paths,  denoted by $\sem{e}_G$, for a property path $e$ and an RDF graph $G$, is shown below. 
{
\small
  \begin{align*}
    \sem{a}_G                        & = \{(s,o) \mid (s,a,o)\in G\},       \\
    \sem{e^-}_G                      & = \{(s,o) \mid (o,s)\in \sem{e}_G\}, \\
    \sem{e_1 \cdot e_2}_G            & = \sem{e_1}_G \circ \sem{e_2}_G,     \\
    \sem{e_1 + e_2}_G                & = \sem{e_1}_G\cup \sem{e_2}_G,       \\
    \sem{e^*}_G                      & = \bigcup_{i \geq 1} \sem{e^i}_G \cup \{(a,a) \mid a \text{ is a term in } G\}, \\
    \sem{e?}_G                       & = \sem{e}_G \cup \{(a,a) \mid a \text{ is a term in } G\}. 
      \end{align*}
      }
Here $\circ$ is the usual composition of binary relations, and $e^i$ is the concatenation of $i$ copies of $e$.



\subsection{Evaluating Property Paths via Automata}\label{sec:automata}
As in the case of the query computing the coauthor reach of M. Stonebraker, one is usually interested in computing all the IRIs that can be reached 
from a starting IRI $u$ by means of a property path expression. Formally, we study the following problem. 
\begin{center}
\fbox{
\begin{tabular}{ll}
{Problem}: & {\sc PPComputation}\\ 
{Input}: & 
Property Path $e$, RDF graph $G$, starting IRI $u$  \\
{Output}: & All IRIs $v$ such that $(u,v) \in \sem{e}_G$
\end{tabular}
}
\end{center}

Alternatively, one may wish to compute the full evaluation $\sem{e}_G$ of pairs connected via a path conforming to $e$. However, this operation is seldom used in practice:
it is not an intuitive query to ask, and when using property paths in SPARQL one usually obtains starting points from other patterns or joins of patterns. Also, computing the full $\sem{e}_G$ is not even supported in all SPARQL systems (for instance Virtuoso allows only property paths with a starting point).  Furthermore, as we will see in the following sections, in the open world setting of Linked Data it is only natural to have a starting point for our search, since it is unrealistic to expect the computation to traverse and manipulate the entire Web graph. This is why we chose to focus on
{\sc PPComputation}.



To solve the {\sc PPComputation} problem, the theoretical literature proposed a simple algorithm based on automata theory. 
To present this algorithm, note first that our property paths are nothing more than regular expressions over the alphabet $\I^\pm = \I \cup \{u^- \mid u \in \I\}$ that contains all IRIs and their inverses. 
Thus, for each property path $e$ we can construct a nondeterministic finite state automaton (NFA) $A_e$ over $\I^\pm$ that accepts the same language as $e$, when viewed as a regular expression. We can now show:
\begin{proposition}[\cite{CMW87,CGLV00,KRRV15}]
\label{prop-query-eval}
{\sc PPComputation} can be solved in $O(|G|\cdot|e|)$ (thus linear in both the size of the graph and the query).
\end{proposition}
The idea is as follows. Let $G$ be an RDF graph, $e$ a property path expression and $u$ an IRI. First, we construct the automaton 
$A_e = (Q_e,\I^\pm,q^0_e,F_e,\delta_e)$ equivalent to the query $e$, where $Q_e$ is the set of states, $q^0_e$ is the initial state, $F$ is the set of final states and 
$\delta_e \subseteq Q_e \times \I^\pm \times Q_e$ is the transition relation. 
Next, from $G$ and $A_e$ we construct the labelled product graph $G \times A_e$ whose nodes come from $\I \times Q_e$, and there is an edge 
from a node $(u_1,q_1)$ to a node $(u_2,q_2)$ labelled with $a\in \I$ if and only if (i) $G$ contains a triple $(u_1,a,u_2)$ and (ii) 
the transition relation $\delta_e$ contains the triple $(q_1,a,q_2)$, that is, if in $A_e$ one can advance from $q_1$ to $q_2$ while reading $a$. Similarly, there is an edge between $(u_1,q_1)$ and $(u_2,q_2)$ labelled with $a^-\in \I^-$ if (i) $(u_2,a,u_1)\in G$ and (ii) $(q_1,a^-,q_2)\in \delta_e$. It is now not difficult to show the following property: 

\begin{lemma}[\cite{CMW87}]
\label{lema-query-eval}
A pair $(u,v)$ belongs to $\sem{e}_G$ if and only if there is a path from $(u,q^0_e)$ to 
$(v,q^f_e)$ in the labelled graph $G \times A_e$, where $q^f_e \in F_e$ is a final state of $A_e$. 
\end{lemma}

We can now solve the {\sc PPComputation} problem by traversing the product graph $G \times A_e$ starting in $(u,q^0_e)$ and returning all the IRIs $v$ such that we encounter a node $(v,q^f_e)$, with $q^f_e\in F_e$, during our traversal.
Thus, in a sense, one can recast the problem of query computation (in a single graph) as the problem of searching for all connected final nodes in the product graph. 
This duality between evaluation and search is a crucial component of our approach for querying multiple graphs on the Web of Linked Data.

\section{Computing property paths over the web}
\label{sec:ComputingPathQueries}

When computing the answer of a property path over the Web, we cannot simply rely on the algorithm outlined in Section \ref{sec:automata}, because this assumes that we have our entire graph in memory, which is not a feasible option for the case of the Web. 
Having a starting IRI $u$ comes in handy here, as we can emulate the algorithm from Section \ref{sec:automata} by dereferencing $u$, retrieving its neighbours in $\adoc(u)$, and continuing from there, thus building a local copy of a portion of the Web graph 
needed to  answer the query.
 
To formalise this, let us define the \emph{Web graph} $G_\text{Web}$ as the RDF graph consisting of the union $\bigcup_{u \in \I} \adoc(u)$ of all the graphs resulting by dereferencing an IRI in $\I$ (i.e. the complete Web of Linked Data). From here onwards we assume that $\adoc(u)$ gives us the neighbours of $u$ in the Web graph (see Section \ref{ss-endpoints} for a discussion of how to deal with the shortcomings of the current Linked Data infrastructure). The problem we are now interested in is solving {\sc PPComputation} above for the graph $G_\text{Web}$, i.e. the problem:
\begin{center}
\fbox{
\begin{tabular}{ll}
{Problem}: & {\sc PP\_over\_the\_Web}\\ 
{Input}: & 
Property Path $e$, starting IRI $u$  \\
{Output}: & All IRIs $v$ such that $(u,v) \in \sem{e}_{G_\text{Web}}$
\end{tabular}
}
\end{center}
Now, although the approach of Section \ref{sec:automata} would require us to do our search over $G_\text{Web}\times A_e$, we can recast {\sc PP\_over\_the\_Web} as finding paths inside a subgraph $G_P\subset G_\text{Web}\times A_e$ which is constructed dynamically by dereferencing IRIs starting at $u$. And although the graph $G_P$ might be much smaller (in fact, we can stop constructing it when we desire), selecting the best algorithm for producing this graph and doing path searching over it is not an obvious task, due to the following issues not occurring in the classical path-finding setting.


First, path-finding algorithms are designed to work with graphs that can be either stored in memory or generated efficiently. In contrast, graph $G_P$ is generated by dereferencing IRIs which involves resolving a number of HTTP requests. The time required to complete a request dominates significantly the time required to carry out any operation performed in memory. Efficient algorithms for this problem should therefore aim at reducing  network requests, a factor that is usually not considered when solving path-finding problems.
The second issue is that here we are interested in more than one solution. As such, 
an algorithm that returns answers incrementally seems to be a more sensible option than one that computes \emph{all} answers prior to returning any.



Next, we discuss how classical path-finding algorithms can be modified to return answers to property path queries over the Web and pinpoint some of their shortcomings in this setting.

\subsection{Depth-First Search}
\label{ssec:DFS}
\emph{Depth-First Search} (DFS) is an easy-to-implement path-finding  algorithm that can be used to solve the {\sc PP\_over\_the\_Web} problem. 
On input a starting IRI $u$ and an automaton $A_e$ over $\I^\pm$, the algorithm begins a search over the graph $G_\text{Web} \times A_e$ starting with the node 
$\init = (u,q^0_e)$, where $q^0_e$ is the initial node of $A_e$. The goal of the algorithm is to look for nodes of the form $(v,q_f)$, with $v$ an IRI and $q_f$ a final state 
of $A_e$; this is commonly known as the \emph{goal} condition of the algorithm. 
At every moment during execution, the algorithm maintains a \emph{search frontier} (or $\Frontier$ list) implemented as a stack. 
At initialisation, the frontier is set to only contain the start node $\init$. In the main loop, a node $s$ is extracted from the frontier and \emph{expanded} by computing its 
neighbours in $G_\text{Web} \times A_e$, by means of the function $\textit{Neighbours}$. All neighbouring goal nodes are returned, and then all neighbours that have not been previously added to the frontier are now inserted at the top of the frontier. The algorithm terminates unsuccessfully if the frontier empties.

\begin{algorithm}
  \fs{7}
\DontPrintSemicolon
\fs{7}
\SetKwData{updaterule}{updaterule}
\SetKwInOut{KwSideEffect}{Effect}

\SetKwFunction{Astar}{A*}
\SetKwFunction{BuildPath}{BuildPath}
\SetKwFunction{Observe}{Observe}
\SetKwFunction{Expand}{Expand}
\SetKwFunction{InsertState}{InsertState}
\SetKwFunction{GoalCondition}{GoalCondition}
\SetKwFunction{InitializeState}{InitializeState}
\SetKwData{Open}{Open}
\SetKwData{Closed}{Closed}
\SetKwData{Null}{null}
\SetKwFunction{Lookahead}{LookAhead}
\SetKwFunction{ExtractMinF}{Extract-Min-$f$}
\SetKwFunction{ExtractMinH}{Extract-Min-$h$}
\SetKwFunction{Insert}{InsertQ}
\SetKwFunction{BestState}{Extract-Best-State}
\SetKwFunction{Dijkstra}{ModifiedDijkstra}
\SetKwFunction{Update}{ReestablishConsitency}
\SetKwFunction{Main}{main}
\SetKwData{back}{back}
\SetKwFunction{Propagate}{Propagate}
\SetKwFor{Proc}{procedure}{}{end}
\SetKwFor{Func}{function}{}{end}
\SetKwFor{Foreach}{for each}{do}{end}
\SetKwFor{Pforeach}{parallel for each}{do}{end}

  \Func{$\textit{Search}(u,A_e)$}{
    $\init \leftarrow (u,q^0_e)$\;
    $\init.\textit{parent} \leftarrow null$\;
    \lIf{$q^0_e \in F_e$}{return $\init$ or add $\init$ to solutions}
    Initialise $\Frontier$ as an empty stack (DFS) or queue (BFS)\;
    Initialise $\Seen$ as an empty set\;
    Insert $\init$ into both $\Frontier$ and $\Seen$\;
        \label{search:init}
    \While{$\Frontier$ is not empty}{
      Extract node $s = (v,q)$ from $\Frontier$ and compute $\textit{Neighbours}(s)$\; 
      \Foreach{$t = (v',q')$ in $\textit{Neighbours}(s)$ that is not in $\Seen$ }{
        $\textit{t.parent}\leftarrow s$\; \label{search:setparent}
        \lIf{$q' \in F_e$}{return $t$ or add $t$ to solutions} \label{search:return}
        Insert $t$ into both $\Frontier$ and $\Seen$
      }
    }
  }
  \Func{$\textit{Neighbours}((v,q))$}{
     Initialise $\Succ$ as an empty set and RDF graph $\Gtemp$ as an empty graph\;
     $\Gtemp \leftarrow \adoc(v)$\; \label{search:extract}
      \Foreach{IRI $a \in \I$ and state $q'$ s.t.  $(q,a,q')$ is in $\delta_e$}{
	  \lForeach{triple $(v,a,v')$ in $\Gtemp$}{
	    Insert $(q',v')$ into $\Succ$
	}

      }
      \Foreach{IRI $a^- \in \I^-$ and state $q'$ s.t.  $(q,a^-,q')$ is in $\delta_e$}{
      \lForeach{triple $(v',a,v)$ in $\Gtemp$}{
	    Insert $(q',v')$ into $\Succ$
	}
      }
      \Return $\Succ$\;  \label{search:expand}
  }
  \caption{Breadth/Depth-First Search}\label{bfsdfs}
\end{algorithm}


A pseudo code for DFS is presented in Algorithm~\ref{bfsdfs}. Note that we need $\Frontier$ to be a stack for DFS (Line~\ref{search:init}). Observe additionally that the algorithm does not return a path but rather a node from which a path can be obtained by following the so-called parent pointers (set in Line~\ref{search:setparent}). Finally, observe that in the context of navigational query answering, computing the neighbours of a node (function $\textit{Neighbours}$) needs IRI dereferencing (set in Line~\ref{search:extract}) which in turn requires network communication, an operation that may take significantly more  time than others carried out by the algorithm, such as data management.

There are three properties of DFS that are important for query answering. First, DFS can be easily modified to return paths incrementally instead of only one path. Indeed, instead of returning in Line~\ref{search:return}, the node just found to be a goal node can be added to a list of solutions. In the same spirit, one can easily adapt DFS to return the first $k$ solutions by introducing $k$ as an additional parameter. Second, DFS is complete for finite graphs: if a goal node is reachable from $init$ then the algorithm eventually retrieves this node. This is important because it guarantees that all solutions to a query are eventually returned. 
Third, the memory footprint of DFS is relatively low. Actually, if the depth of the node on top of the stack is $k$ and the maximum branching factor (number of neighbours of a node) is $b$, then the size of $\Frontier$ is $O(kb)$.

To see how DFS works when solving {\sc PP\_over\_the\_Web}, let us consider the first steps taken when processing the property path $(\texttt{dc:creator}^- \cdot \texttt{dc:creator})*$, with the starting IRI $\MS$, which was presented  
in the introduction. First, let $A_e$ be the following NFA: 
\vspace{-5pt}
\begin{center}
\begin{tikzpicture}[>=stealth',auto,node distance=2cm]
  \node[initial,state,accepting,minimum size=0pt]   (q0)               {$q_0$};
  \node[state,minimum size=0pt]           (q1) [right of=q0] {$q_1$};

  \path[->]  (q0)  edge [bend left] node {\scriptsize $\texttt{dc:creator}^-$} (q1);
  \path[->]  (q1)  edge [bend left] node {\scriptsize $\texttt{dc:creator}$} (q0);
\end{tikzpicture}
\end{center}
\vspace{-5pt}
As explained in Lemma \ref{lema-query-eval}, the starting node for our search is $(\MS,q_0)$. In the first iteration we extract this node from $\Frontier$, dereference the IRI $\MS$, obtaining, amongst others, the triple (\texttt{inTods:StonebrakerWKH76}, \texttt{dc:creator}, \MS). This will allow us to add to our frontier the node $(\texttt{inTods:StonebrakerWKH76},q_1)$ and we proceed similarly for other triples in $\adoc(\MS)$. For the second iteration, let us assume $(\texttt{inTods:StonebrakerWKH76},q_1)$ is at the top of the stack. When expanded, DFS will lookup $\adoc(\texttt{inTods:StonebrakerWKH76})$ and 
retrieve, amongst other things, all nodes connected to \texttt{inTods:StonebrakerWKH76} by means of a label $\texttt{dc:creator}$. This, in particular, yields all 4 authors of this paper, but $(\MS,q_0)$ is not added to $\Frontier$ because 
it was already in $\Seen$.
For the next iteration DFS takes one of these nodes, say $(\texttt{G.Held},q_0)$, expands them again, obtaining all papers of G. Held; the next iteration expands one of these papers, adds all the authors to the list of answers; and so on.

Algorithm~\ref{bfsdfs} implements a \emph{loop detection} by preventing the insertion of a previously seen node to $\Frontier$. This is important to guarantee that the algorithm 
terminates over a finite graph and that the answers are complete. In our case this implies that we are looking for simple paths, albeit not in the RDF graph but in the product 
graph $G_\text{Web} \times A_e$. In practice this implies that our algorithm looks for paths where the same IRI may be repeated at most a number of times equivalent to the states of 
the expression automata $A_e$. Completeness of DFS in our context follows from a simple pumping argument and the fact that property paths are regular expressions over 
$\I^\pm$. 



The most notable drawback of DFS 
is that there is no guarantee on solution quality, and solutions with much shortest paths may be missed. For instance, in the query above it will return the co-authors of G. Held, which are at distance two or more from M. Stonebraker, before returning the other authors of the paper \texttt{inTods:StonebrakerWKH76}. 
In practice this means that we would need many more HTTP requests to retrieve subsequent solutions, which in turn means more time to compute answers.

\subsection{Breadth-First Search}
\label{ssec:BFS}

To alleviate the drawbacks of DFS, one could consider instead using 
{\em Breadth-First Search} (BFS), another complete search algorithm that is guaranteed to find shortest paths. BFS is similar to DFS in most aspects: it keeps a search frontier (i.e. the $\Frontier$ list) during execution and in each iteration it extracts a node from the frontier and then expands it. The most important difference is that BFS, instead of always expanding the deepest node in the frontier, it always expands the shallowest one. At the algorithmic level, BFS can be obtained from DFS by simply changing underlying data structure for \Frontier to a FIFO queue instead of a stack. As such, successors of a node are always added at the end of the queue, and therefore a shallow node is always selected for expansion. 

However, BFS also suffers from an important drawback in our context: BFS has the potential of needing many iterations to find a first solution to the problem. Indeed, assume once again that a node has at most $b$ neighbours, and imagine that the shortest path in the search graph has $k$ edges. Then, \emph{all} nodes that are reachable in less than $k$ edges are added to \Frontier which means that $O(b^k)$ iterations are needed before such a path is found. 

\subsection{Issues with BFS and DFS}
Both BFS and DFS have issues with some queries. Consider for example the following query, starting with 
M.~Stonebraker: 
 $$(\verb+dc:creator+^- \cdot \verb+dc:creator+)^* \cdot \verb+dc:creator+^- \cdot \verb+rdfs:label+,$$
that is, intuitively we want to retrieve the papers written by a co-author of M. Stonebraker, or by a co-author of some of his co-authors, and so on.  Furthermore, take the realistic assumption that there are hundreds of IRIs connected via $\verb+dc:creator+^-$ with M. Stonebraker (indeed, Stonebraker's DBLP entry, as of the writing of this paper, contains 298 papers).

Let us now focus on what BFS does with this query. It will first dereference the IRI for Stonebraker, adding the IRI of each of his 298 papers to \Frontier. Then, it will dereference each of these IRIs, which requires 298 requests over the network. When each of these IRIs are expanded, we add to \Frontier the co-authors of Stonebraker. Only after all the IRIs for Stonebraker's papers are expanded, it will expand the IRI of one Stonebraker's coauthors, and, immediately will find a solution path.

Waiting for 298 HTTP requests before obtaining the first answer is not sensible: in this case only three requests are needed to find the first answer. Indeed, starting from Stonebraker's IRI, we just choose the IRI for \emph{one} of Stonebraker's papers, we expand such an IRI, from where we choose the IRI of \emph{one} of his co-author's. After dereferencing the latter IRI we find the first solution.

DFS has different yet important issue with this very same query. To find a first
solution, DFS actually does the minimum amount of effort, dereferencing the
minimum number of IRIs, as described above. The issue appears when looking for
the answers that follow the first. Because the focus of DFS is depth, when executed over DBLP, the
$5^{th}$ answer of our query has length 6, the next 4 answers have length 12, and the
following ones 32 and up.
This implies that DFS will incur in more computation time to retrieve these answers, as well as more http requests. Moreover, returning these lengthy paths first does 
not seem intuitively right, as we normally want to display simpler, shorter paths first.  Indeed, it is not hard to contrive examples in which the length of solutions increases much faster than in our examples, even when many shorter solutions exist. 

What we need is a good balance between execution time and solution quality. In our example, a sensible way to proceed would be to take the IRI for the first paper, look at its authors, list them, and then proceed likewise with the second paper. This balance 
has been studied in the area of Heuristic Search, 
for many years, producing algorithms that are guided by a heuristic function $h$,  that is such that $h(s)$ estimates the cost of a path from $s$ to a goal node. Expansions are significantly (usually, exponentially) reduced as one improves the ``quality'' of $h$. 
Next we discuss 
the challenges of using of using heuristic search over Linked Data.


\section{AI search to the rescue}
\label{sec:AISearch}

A* is one of the most simple and well-studied heuristic algorithms capable of solving path search problems like the one we described in the previous sections. In this section we study how to apply it to the problem of answering property paths over the Web.

The main difference between A* and the algorithms described earlier is that the search frontier is a priority queue where the priority is given by $f(s)$, a function that estimates the cost of a solution that passes through $s$ \cite{HartNR68}. A high-level description of A* is as follows. At initialisation, the initial node is added to the \Frontier queue. A* now repeats the following loop: first, it extracts a node with the highest priority from \Frontier. It returns $s$ if it is a goal state; otherwise, it expands $s$ to obtain its neighbours, adds them to \Frontier and continues execution. 
Next we give a formal description of A*.

The search graph of A* is 
\emph{implicitly} described by (1) a start node $\sstart$; (2) a set of actions $Act$; (3) a partial successor function $Succ$, such that $Succ(a,s)$, if defined, returns a set $S$ of successor nodes; (4) a goal condition, which is a boolean function over nodes---\emph{goal nodes} are those nodes for which this function returns true; (5) a non-negative cost function $c$ between successor nodes.  The objective of the algorithm is to find a path from $\sstart$ to a goal node.

An additional argument required by A* is a heuristic function $h$, which is a non-negative function over nodes such that $h(s)$ is an estimate of the cost of a path that starts in $s$ and reaches a goal node. The heuristic is key to the performance of A*. An empirically well-known fact is that as $h$ is more accurate, time savings can be very big because expansions are significantly reduced. It can be proven that when $h$ is \emph{admissible}, that is, for every $s$ it holds that $h(s)$ does not overestimate the cost of any path from $s$ to a goal node, then A* finds a minimum-cost path from \sstart to a goal node.


\begin{algorithm}[t]
  \caption{The A* Algorithm}\label{alg:astar}
\fs{7}
\DontPrintSemicolon
\fs{7}
\SetKwData{updaterule}{updaterule}
\SetKwInOut{KwSideEffect}{Effect}

\SetKwFunction{Astar}{A*}
\SetKwFunction{BuildPath}{BuildPath}
\SetKwFunction{Observe}{Observe}
\SetKwFunction{Expand}{Expand}
\SetKwFunction{InsertState}{InsertState}
\SetKwFunction{GoalCondition}{GoalCondition}
\SetKwFunction{InitializeState}{InitializeState}
\SetKwData{Open}{Open}
\SetKwData{Closed}{Closed}
\SetKwData{Null}{null}
\SetKwFunction{Lookahead}{LookAhead}
\SetKwFunction{ExtractMinF}{Extract-Min-$f$}
\SetKwFunction{ExtractMinH}{Extract-Min-$h$}
\SetKwFunction{Insert}{InsertQ}
\SetKwFunction{BestState}{Extract-Best-State}
\SetKwFunction{Dijkstra}{ModifiedDijkstra}
\SetKwFunction{Update}{ReestablishConsitency}
\SetKwFunction{Main}{main}
\SetKwData{back}{back}
\SetKwFunction{Propagate}{Propagate}
\SetKwFor{Proc}{procedure}{}{end}
\SetKwFor{Func}{function}{}{end}
\SetKwFor{Foreach}{for each}{do}{end}
\SetKwFor{Pforeach}{parallel for each}{do}{end}

\Proc{A*}{
    $\Closed\leftarrow \text{empty set}$\;
    $\Open \leftarrow \text{empty priority queue ordered by $f$ attribute}$\;
    $g(\sstart)\leftarrow 0$\;
    $f(\sstart)\leftarrow h(\sstart)$\;
    Insert $\sstart$ into $\Open$\;
    \While{$\Open\neq\emptyset$}{
      Extract $s$ from $\Open$\; \label{alg:astar_extract}
      \If{$s$ is a goal node}{
        $\Return$ $s$ or add $s$ to list of solutions  \label{alg:return}
      }
      $\Expand(s)$ \label{alg:astar_call_expand}
    }
  }
  \Proc{\Expand$(s)$}{
    Insert $s$ into \Closed\;
    \Foreach{$a$ in $Act$ such that $\Succ(a,s)$ is defined}{
     \Foreach{$s'$ in $\Succ(a,s)$}{
      $t\leftarrow s'$\; \label{alg:generation}
      \If{$t$ is not in $\Seen$}{
        Add $t$ to $\Seen$ \; \label{alg:store}
        $g(t)\leftarrow \infty$
      }
      $cost\leftarrow g(s)+c(s,t)$\;
      \If{$cost<g(t)$}{
        $g(t) \leftarrow cost$\;
        $f(t) \leftarrow g(t)+h(t)$\;
        $parent(t) \leftarrow \angled{s,a}$\;
        \If{$t$ is a goal and $f(t)\leq f(top(\Open))$}{
          \Return $t$ or add $t$ to list of solutions \label{alg:secondreturn}
        }
        \lIf{$t\not\in\Open$}{
          Insert $t$ in $\Open$
        }
        \lElse{
          Update priority of $t$ in $\Open$
        }
      }
    }
    }
  }
  
\end{algorithm}

Algorithm~\ref{alg:astar} shows a pseudo-code for A*. The priority function is defined as $f(s)= g(s)+h(s)$, where $h$ is the heuristic function defined above and $g(s)$ is the cost of the best path found so far towards $s$. In an implementation of A*, a hash table is used to store nodes that have been generated in an expansion (cf. Line~\ref{alg:store}), and $parent$-, $g$-, $h$-, and $f$- values are stored as properties of $s$. 

A final and important observation is that A* can be easily modified to return a sequence of answers, instead of a single one. In this case, we simply modify the return statement in Line~\ref{alg:return} by something that adds $s$ to a list. 

\smallskip
\noindent
\textbf{Using A* for computing property paths}.
Let us show how we use A* to solve {\sc PP\_over\_the\_Web}. 
That is, given as inputs a property path $e$ and a starting IRI $u$, we look for all $v$ such that 
$(u,v) \in \sem{e}_{G_\text{Web}}$. Let $A_e = (Q_e,\I^\pm,q^0_e,F_e,\delta_e)$ be the automata over $\I^\pm$ that is equivalent to $e$. 
Recall that (see Lemma \ref{lema-query-eval} and Section \ref{sec:ComputingPathQueries}) we can reduce this problem to searching for all nodes $(v,q_f)$, for an IRI $v$ and a state $q_f \in F_e$, over the graph $G_P \subset G_\text{Web} \times A_e$ 
such that there is a path from $(u,q^0_e)$ to $(v,q_f)$.
In turn, this problem can be seen as an A*  description where 
\textit{(1)} the start node $\sstart$ is $(u,q^0_e)$;
\textit{(2)} the set $Act$ of actions corresponds to IRIs in $\I^\pm$;
\textit{(3)} the partial successor function $\Succ$ corresponds to the edges of $G_\text{Web} \times A_e$, that is, 
if $s=(u,q)$, we say  $(u',q')\in\Succ(a,s)$ if both $(u,a,u')\in adoc(u)$, and $(q,a,q')$ is a transition in $A_e$, or if both 
$(u',a,u)\in adoc(u)$, and $(q,a^-,q')$ is a transition in $A_e$;
\textit{(4)} a node $s = (v,q)$ is a goal if $q \in F_e$; and 
\textit{(5)} the cost function is $1$ for each pair of nodes connected by $\Succ$.


\smallskip
There is an important subtlety that distinguishes our algorithm from classical A* applications. 
Just as in the case of BFS and DFS, the successors of $(u,q)$ must be obtained by dereferencing an IRI (using, for example, the function $\textit{Neighbours}$ from Algorithm \ref{bfsdfs}). This again means that the most costly operation 
is the expansion of new successor nodes, and as such any implementation of A* must try their best to find a way of reducing this bottleneck. We explain how to do this in Section \ref{sec-parallel}.
But before, let us see how to choose a good heuristic function in our scenario. 

\subsection{A Heuristic for Navigational Queries}
Heuristic functions are essential for the performance of A*. We also want A* to be optimal, so our heuristic must be admissible, that is, it should not overestimate the cost of path to a goal node. 

Let $A$ be an automaton over $\I^\pm$. Our heuristic for this problem is defined as follows: for all nodes $(v,q)$ 
over $\I \times Q_e$, where $Q_e$ are the states of $A_e$, we define $h((v,q))$ as the 
minimum distance from $q$ to a final state of $A_e$ (and as $\infty$ if no path from $q$ to a final state exists).  
To illustrate our heuristic consider Figure~\ref{fig:sampleAnswer}, corresponding to the automaton of the query for papers of the coauthor reach of M. Stonebraker introduced in Section 3.3. 
Then we define $h(u,q_1)=2$, $h(u,q_0)=1$, and $h(u,q_2)=0$, for every $u \in \I$. 
Usually $h(u,q)$ is implemented as a simple lookup in a table. Given an automaton $A$ we denote the heuristic defined as described above by $h_{A}$.

Our heuristic $h_{A_e}$ is admissible for each property path $e$, as long as $A_e$ is the minimum NFA for $e$. To see this, 
note that the minimum number of actions required to reach a goal node from node $(u,q)$ cannot exceed the number of edges of a shortest path between the automaton state $q$ and a final state. This is because each successor  $(u',q')$ of node $(u,q)$ must be such that there is an edge between $q$ and $q'$ in the automaton's graph. 

\begin{figure}
\begin{center}
\begin{tikzpicture}[>=stealth',auto,node distance=1.7cm]
  \node[initial,state,minimum size=0pt] (q0)      {$q_0$};
  \node[state,minimum size=0pt]         (q1) [above=0.6cm of q0]  {$q_1$};
  \node[state,accepting,minimum size=0pt]         (q2) [right=2.2cm of q0] {$q_2$};

  \path[->]  (q0)  edge [bend left] node[left] {\scriptsize $\texttt{dc:creator}^-$} (q1);
  \path[->]  (q1)  edge [bend left] node[right] {\scriptsize $\texttt{dc:creator}$} (q0);
  \path[->]  (q0) edge  node[above] {\scriptsize $\texttt{dc:creator}^-$} (q2);
\end{tikzpicture}
\end{center}
\vspace*{-10pt}
\caption{An automaton finding papers of the co-authors of M.~Stonebraker.}\label{fig:sampleAnswer}
\label{fig-aut}
\end{figure}
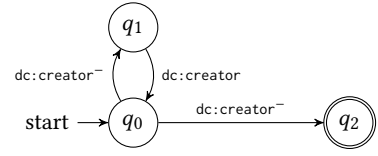






\subsection{Theoretical Guarantees}
A well-known property of A* is that is finds cost-optimal (i.e., shortest) paths. Here we provide an optimality result of the same sort. 
Now, because the function $\adoc(u)$ is not necessarily guaranteed to return all triples containing $u$, we cannot show optimality 
over the entire Web, but rather only over the graph we have already discovered, that we denote by $\Gastar$. 

Formally, given an execution of A*, the labelled graph $\Gastar \subset G_\text{Web} \times A_e$ contains the edge $(s,a,s')$ iff (1) $s\in Closed$, and (2) $s'\in Succ(a,s)$. 
Notice that this correspond to the product of $A_e$ with the graph that contains all triples present in any of the documents that have been retrieved so far in our computation. 
The following is an optimality result both for BFS and A* run with our property-path heuristic.



\begin{theorem}\label{thm:optimality}
  Let \Gastar be defined as above from a run of A* that has returned $N$ answers with either $h=h_{A_e}$ or $h=0$. Let $\pi_k$ be path found to the $k$-th solution found by A*, for any $k\in \{1,\ldots,N\}$. Furthermore, let $c_k$ be the length of the $k$-th shortest path from \sstart to any goal state over \Gastar. Then the cost of $\pi_k$ is $c_k$.
\end{theorem}
\begin{proof}[sketch]
Let $\Gastar^i$ denote \Gastar right after the $i$-th solution has been returned. We prove by induction that the $i$-th solution found by A* would be the first solution found by A* if we were to mark as non-goals all solutions found prior to the $i$-th solution. Now we use the fact that the heuristic is admissible and thus the solution found is the $i$-th optimal over $\Gastar^i$.
\end{proof}

In practice, as we see later on, BFS runs slower than A* with our heuristic. Interestingly, we can prove that A* is better in the sense that BFS has to expand at least as many nodes as A*.
\begin{theorem}\label{thm:dominance}
  Let $(u,e)$ be an IRI and a property path. Then every node expanded by A*, used with $h_{A_e}$, is also expanded by BFS.
\end{theorem}
\begin{proof}
We observe that $h_{A_e}(s)>0$ for every non-goal state $s$. The result now follows from Theorem~7 in \cite{Pearl84}.
\end{proof}

\section{Optimising query execution}
\label{sec:Optimisations}



In this section we provide several optimisations to the base algorithms presented in Section \ref{sec:ComputingPathQueries} and Section \ref{sec:AISearch}. We start by describing how parallel expansions can be used in order to reduce the execution times of our search algorithms.
We then explain what are the current shortcomings of the Linked Data infrastructure and propose a way to overcome them using endpoints.
Finally, we discuss a way of tweaking the heuristic used in A* in order to both avoid unnecessary network requests and find answers sooner.


\subsection{Parallel Expansions}
\label{sec-parallel}
All of our search algorithms function in such a way that they select a set of nodes which will serve as the starting point in the next iteration, and then start the search from these nodes one by one.
An issue with this is that a request over the network---which on average takes more than a second--- is needed per each dereference. 

Instead of expanding one node at a time, our algorithms can benefit greatly from expanding multiple ones in parallel.
More specifically, we modify the algorithm to extract up to $k$ of nodes that could be at the top of the \Frontier, and expand them in parallel. $k$ is now a paramenter of the algorithms which can be understood as a \emph{degree of parallelism}.
To obtain $k$-BFS and $k$-DFS, we modify Algorithm~\ref{bfsdfs} such that Line \ref{search:extract} deals with up to $k$ top-valued nodes from $\Frontier$, and neighbours are computed for them in paralell. Similarly, $k$-A* is obtained by extracting up to $k$ nodes with the highest $f$-values from \Frontier in Line \ref{alg:astar_extract} of Algorithm \ref{alg:astar}, and expanding them all in parallel in line \ref{alg:astar_call_expand}. In all 3 algorithms, after all successors are computed, we add them all together to \Frontier, in the same order that we would have, had the nodes been expanded sequentially. It is then not hard to see that optimality (Theorem~\ref{thm:optimality}) still holds for $k$-A* (and $k$-BFS).

In Section \ref{sec:Evaluation} we show that, depending on the degree of parallelism, the computation is sped up  tenfold in some instances. 

\subsection{Using The Endpoint Infrastructure}\label{ss-endpoints}
\label{sec-endpoint}
The evaluation algorithms presented in previous sections rely on the dereferencing mechanism of Linked Data and work under the assumption that when a specific IRI is dereferenced, we obtain all the triples mentioning such an IRI which reside on the server we are using. Unfortunately, it was shown repeatedly that this is generally not the case when working with Linked Data \cite{HoganUHCPD12,HoganG14}, which can lead to incomplete answers since many triples containing the dereferenced IRI might not be returned. This is particularly problematic when working with inverse links, as it is estimated that publishers include only about a half of the triples where the requested IRI appears as the object \cite{HoganUHCPD12}.

Many Linked Data providers also set up public SPARQL endpoints where users can query the dataset, so we can partially alleviate the lack of Linked Data infrastructure  
by relying on public SPARQL endpoints together with Linked Data. When evaluating property paths over Linked Data, we combine the two approaches and, each time we dereference an IRI, we also query the appropriate endpoint in order to obtain the triples mentioning the said IRI. Furthermore, we query the endpoint only asking for links in the appropriate direction. For instance, if our property paths needs to traverse the \texttt{author} edge forwards starting from an IRI \texttt{start}, we ask the query \texttt{SELECT ?x WHERE \{start author ?x\}} to the appropriate endpoint and similarly for the backwards edges.

\begin{figure*}
  \centering
  \begin{subfigure}
    \centering
    \includegraphics[width=.3\linewidth]{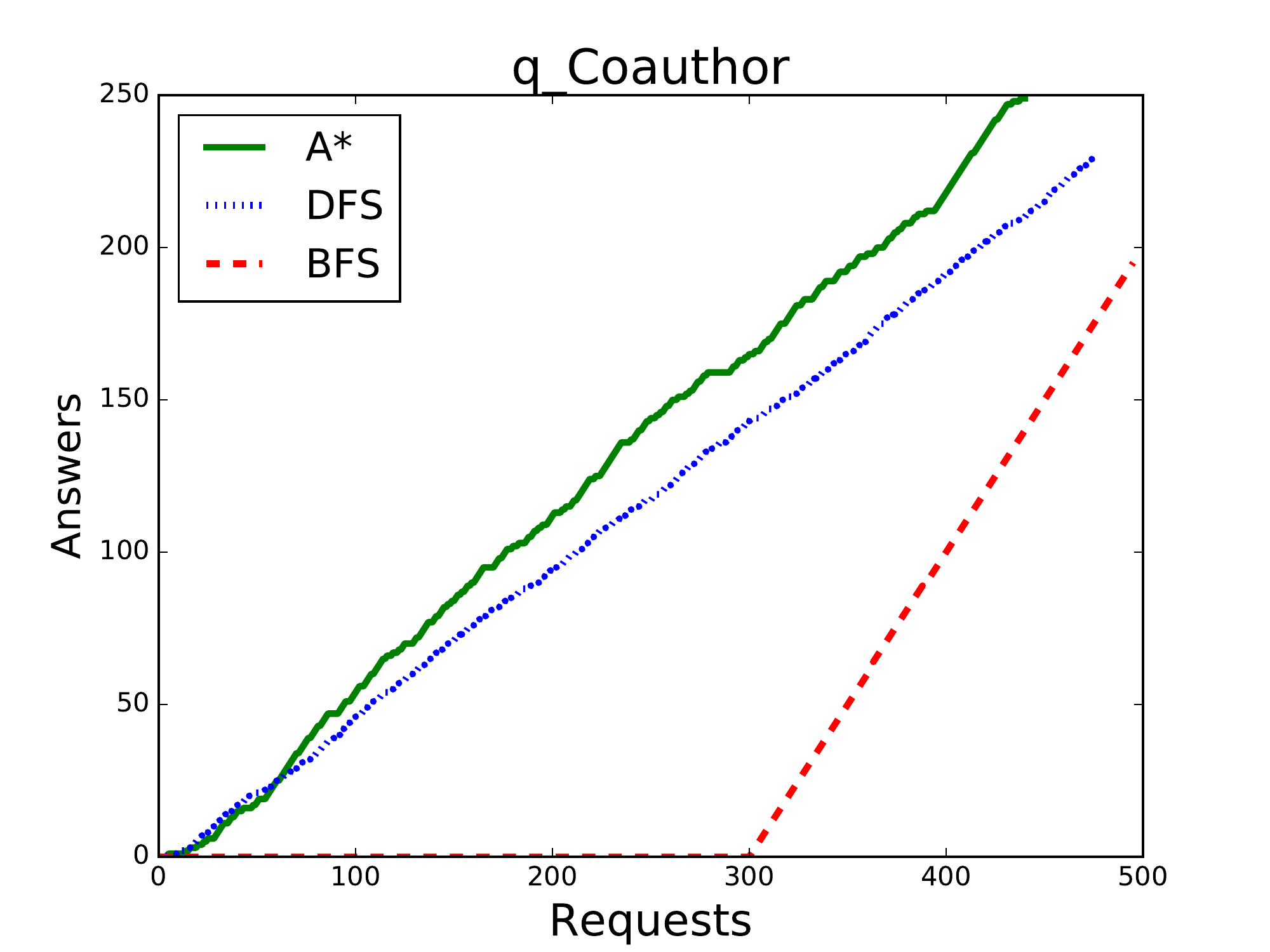}
  \end{subfigure}%
  ~
  \begin{subfigure}
    \centering
    \includegraphics[width=.3\linewidth]{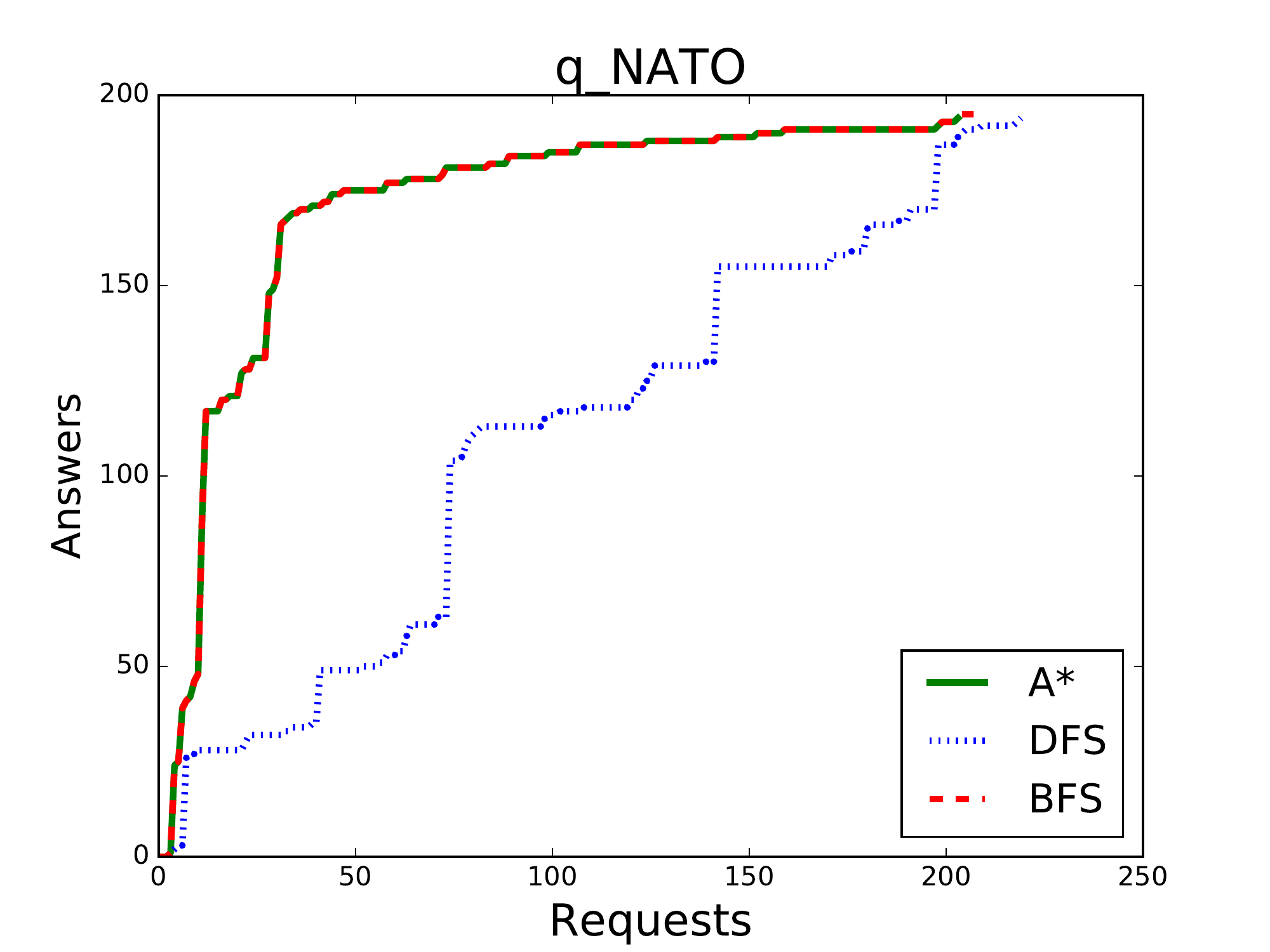}
  \end{subfigure}%
   ~  
  \begin{subfigure}
    \centering
    \includegraphics[width=.3\linewidth]{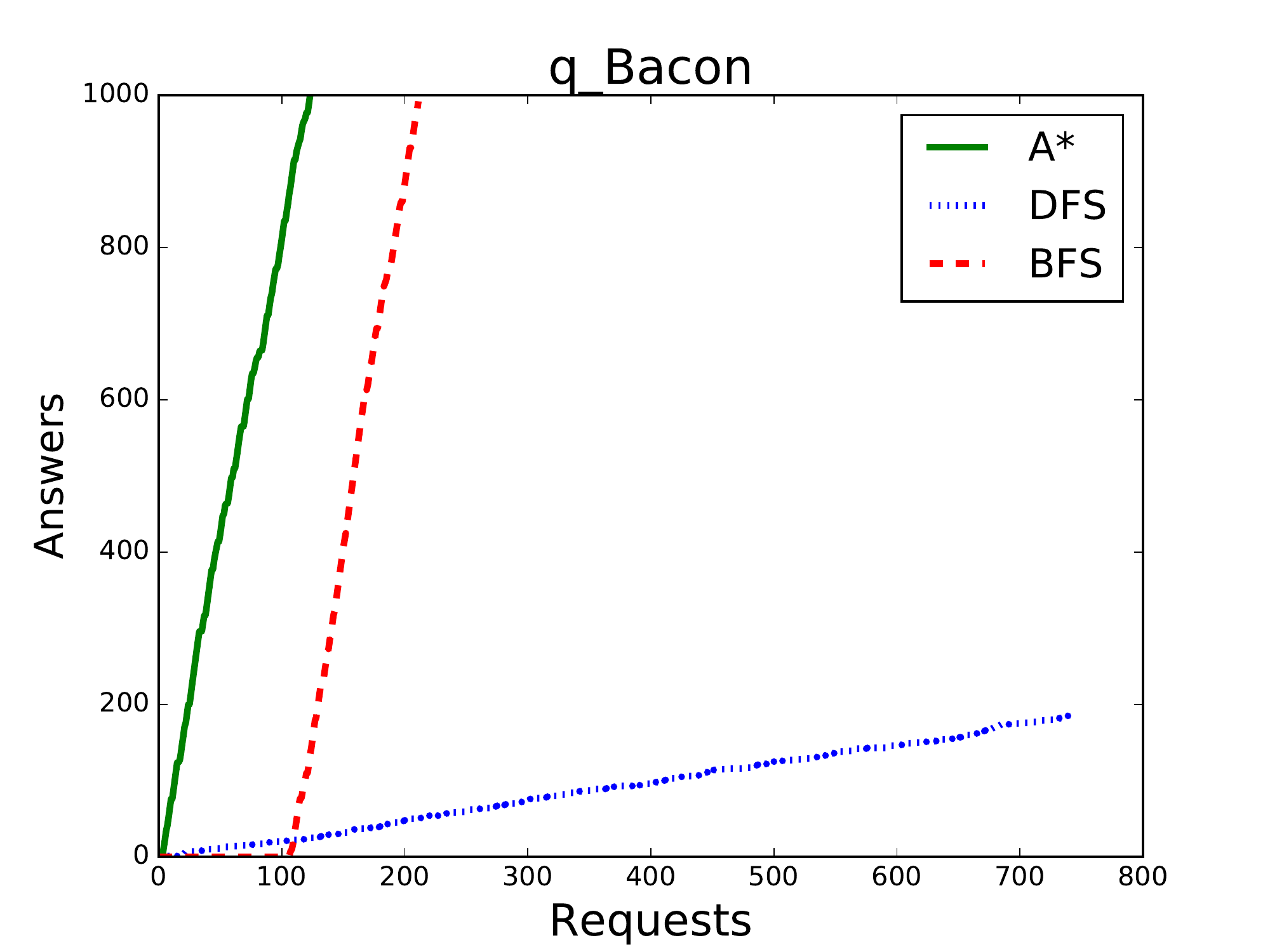}
  \end{subfigure}%
  \vspace{-15pt}
   \caption{A* minimises the requests needed to obtain answers of \qA, \qC and \qB}
  \label{fig:simple_ans-time}

\end{figure*}

\subsection{Minimising Network Requests}\label{ss-optheuristic}
We have argued above that dereferencing is an expensive operation. 
When A* is modified to find multiple answers, as we proposed above (i.e., by simply adding solutions to a list), some expansions may be carried out sooner than we would want, leading to unnecessary dereferencing. 
Indeed, our heuristic assigns the value $0$ to any node of the form $(u,q_f)$ with $q_f$ a final state (because the distance to a final state is $0$). 
Assuming $q_f$ has outgoing transitions, the standard A* algorithm would prioritise the expansions of those nodes over any other node with the same $f$ value, an operation that intuitively would take us farther from the goal.  


We can postpone these expansions by using a slightly different heuristic, defined as follows. Let $A = (Q,\Sigma,q_o,F,\delta)$ be an NFA. The \emph{pathmax} distance $\hat d(q,q')$ between two states is defined as 
$\hat d(q,q') = 1 + \min_{q \mid \delta(q,a) \text{ is defined}} d(q,q')$, if $\delta(q,a)$ is defined for at least some $a \in \Sigma$, or $\hat d(q,q') = \infty$ otherwise; and where $d$ is the usual graph distance between $q$ and $q'$. 
Then the pathmax heuristic $\hat h_A$ with respect to $A$ is defined as $\hat h_A((u,q)) = \min_{q_f \in F} \hat d(q,q_f)$, that is, the minimum pathmax distance from $q$ to any final state of $A$.  
This is a standard technique used by search algorithms in which a node may have to be re-expanded \cite{Russell92,Korf93}. Interestingly, one can see that the pathmax distance $\hat d$ coincides with the usual distance in all states of the automata except for the final states. We 
avoid early re-expansion of nodes with final states because the heuristic for them now corresponds to $1$ plus the minimum of the heuristic value of the neighbours of this state. 

\section{Experimental Evaluation}
\label{sec:Evaluation}

%
%
%
%
%

In this section we evaluate how an implementation of A* algorithm performs when executing property path queries over a real Web environment. 
To establish a baseline, we also compare the performance of A* with BFS and DFS, the only other algorithms proposed so far in the literature. 
%
We begin by presenting our experimental setup, the queries, and then compare how A* fares against BFS and DFS, showing that A* outperforms the other two algorithms on a regular basis.
Next we investigate the impact of parallel requests on our algorithms. 
As we see, adding parallelism reduces total runtime for all three algorithms, with A* remaining the 
most consistent. Interestingly, all algorithms tend to look more alike as more and more parallel requests are allowed.  
Finally, we discuss some real-world examples of the paths retrieved by our algorithms. 

\subsection{Experimental Setup}

We selected $11$ navigational queries that are inspired by previous benchmarks (see e.g. \cite{GubichevBS13,RSV15}). These queries 
are representative of several different features of property paths, ranging from easy, fixed-depth ones to queries using multiple 
star operations which are much harder to evaluate. Our queries target one or more of the following Linked Data domains: 
YAGO, a huge knowledge base extracting data from Wikipedia and various other sources \cite{YAGO}; DBPedia \cite{dbpedia}, one of the central datasets 
of the Linked Data initiative that also originates from Wikipedia; Linked Movie Database, the best known semantic database for movie information \cite{lmdb}; 
and the Linked Data domain of DBLP. 
Our implementations will always use the optimisation techniques presented in Section \ref{ss-endpoints} and Section \ref{ss-optheuristic}, while we assess the benefits of parallelism (Section \ref{sec-parallel}) separately.
As an example, below are $3$ of the $11$ queries we use. 
For complete details of the queries, results of all our runs, and implementation of our algorithms, please refer to our online appendix \cite{appendix}.

\noindent
\qA: The property path $(\texttt{coauthor}^-\cdot \texttt{coauthor})^*$ in the DBLP dataset, starting from the IRI of M. Stonebraker.  This property path looks for the IRI of all authors that are related to M. Stonebraker on DBLP, by a co-authorship path of arbitrary length. 

\noindent
\qC:  A property path that selects all places that host an entity dealing with a NATO member state, according to YAGO. 

\noindent
\qB: A property path that looks for the IRIs of actors having a finite Bacon-number\footnote{An actress has Bacon number 1 if she acted in the same movie as Kevin Bacon, and Bacon number $n$ if she acted with someone with Bacon number $n-1$.},  
and that navigates using links and/or IRIs present in any of YAGO, DBPedia or Linked Movie Database. This is an interesting query, as currently the only way to evaluate it is by means of our Linked Data approach (see \cite{BaierDRV16} for a discussion).



To assess our algorithm we use two indicators: the number of HTTP requests made to compute a fraction of the answers, and the time needed to compute them. In both cases we want to minimise the number of requests, or the amount of time needed to produce the answers. 
We note that the number of requests is a much better indicator on how the algorithm works: because HTTP requests take considerably more
time than all the other operations, the total time of computing our queries is essentially given by the number of requests performed by the algorithm.
This also rules out the dependence on parameters which we have no control over, such as the Internet traffic, or the availability of servers providing us with data. Thus, by focusing on requests we ignore latency differences that may persist even after taking several runs of the same query.

All experiments were run without an access to the data locally, relying solely on the Web infrastructure to retrieve the data needed at each step of the computation. 
The experiments were run on a Manjaro Linux machine with a i5-4670 quad-core
processor and 4GB of RAM. To avoid flooding servers with requests we only ran our search until we either found 1\,000 answers, retrieved more than 100\,000 triples from the server, or reached a 10-minute time limit. Each experiment was ran 10 times, and since the results were largely equivalent, we report the numbers of the latest execution. The source code for running the experiments is
available at \cite{appendix}.

\subsection{Heuristic Search Against BFS and DFS}

The general conclusion of our experiments is that A* both requires fewer requests and is faster than BFS and DFS. 
Before reporting our results in full, let us examine the runs of queries \qA, \qB and \qC presented above.   
Figure \ref{fig:simple_ans-time} shows the number of requests needed to compute a fraction of the total answers available 
for these queries. 
In particular, we see that both A* and DFS are the best choice for the query \qA, because they produce more answers using fewer HTTP requests 
(even though the quality of the answers produced by A* is arguably better -- see below). On the other hand, BFS requires  around 300 expansions 
to start producing answers, which results in a much slower throughput altogether. 
Next, both A* and BFS are the best choice for the query \qC. This is again expected, because this query requires less navigation and more shallow exploration. 
It is interesting to see that in this case A* really simulates the optimal BFS search. 
On the other hand, DFS wastes a lot of time exploring long paths and obtaining ``deep" answers.
Finally, in the case of \qB, A* is shown to strictly beat both BFS and DFS. In the case of BFS, this is mostly because A*'s heuristic allows a finer control on which links to explore, and the main detractor for DFS is that it starts exploring initial links which often require many requests before encountering a solution.

\smallskip 
\noindent
\textbf{Full results}. For reasons of space, we cannot report the remaining 8 experiments with the same detail, so instead we do the following. 
For each query, we analyse the complete behaviour of the answers vs.\ request and answers vs.\ time curves. We say that an algorithm \emph{dominates} the others if it is such that it returns at least as many answers as any other for 80\% of the range of requests (or time) for which we evaluate them.
For example, in Figure \ref{fig:simple_ans-time} we see that $A^*$ dominates the other algorithms for queries \qA and \qB, while in the case of \qC both $A^*$ and BFS dominate.
The total number of times each algorithm dominates (out of 11 queries) is shown below, for both the answers vs.\ request and answers vs.\ time curves (full details are found in our online appendix \cite{appendix}). Once again,  $A^*$ remains the most consistent option.
\begin{center}
  \begin{tabular}{l|lll}
    \hline
    Measure & A*  & BFS & DFS  \\ \hline
    Requests v/s Answers       &  11     & 3  & 4    \\
    Time v/s Answers         &  11     & 3   & 4   \\
  \end{tabular}
  \end{center}


\subsection{The Effect of Parallel Requests}
\label{subsec-experiments-parallel}
Next, we test the effect of allowing parallel requests in our algorithms, as presented in Section \ref{sec-parallel}. 
This optimisation goes a long way into tackling the slow latency of Web requests, one of the main problems of querying over the HTTP protocol. 
%
Indeed, HTTP requests are such an important bottleneck in our algorithm that allowing parallel requests essentially means parallelising the entire algorithm. Issuing parallel requests also soften up high latency pockets or temporary network problems.
Moreover, we can also expect the algorithms to be accelerated even further when the number of allowed requests is increased, simply because more requests essentially means more parallel instances of our algorithm and even more softening power. 
The other interesting observation is that, as we allow more parallel requests in our algorithms, all of A*, BFS and DFS start to look alike,
and in fact it is easy to see that all three algorithms are essentially equivalent in the limit where we issue an infinite number of requests at the same time.

To empirically test these observations, we issued new live runs of the 11 queries described in the previous sections, but this time using 
parallel versions of A*, BFS and DFS. To see the impact on the number of parallel threads allowed, we report experiments with 
a maximum of 10, 20, and 40 parallel threads. Before reporting the full results, let us start with comparing the results of the algorithm with no parallelism against the
one with 20 parallel threads. Figure \ref{fig-parallel} shows the time needed to compute the answers of query \qA, for all three algorithms on their non-parallel version 
and on their parallel version with a maximum of 20 threads. As we see, the time needed to compute the same amount of answers decreases by almost tenfold  in all three cases. 
Moreover, the parallel version of BFS now behaves almost as A* and DFS when computing the first 300 answers (it then reaches a stalemate because all shallow 
answers have already been discovered). 

\begin{figure}[t]
  \centering
  \includegraphics[width=.75\linewidth]{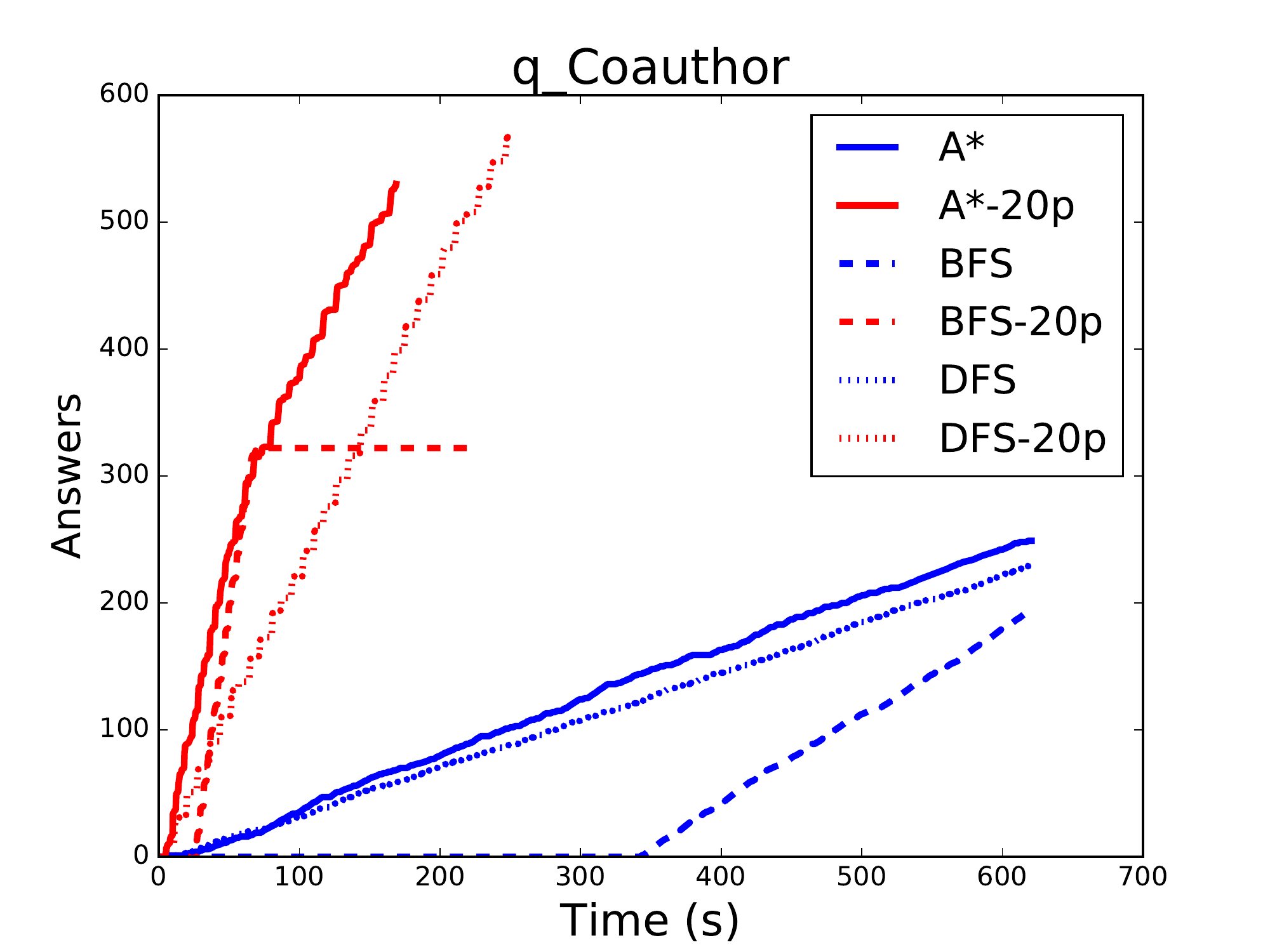}
  \vspace{-10pt}
  \caption{Answers over time on \qA. The parallel versions (in red) are much faster than the non-parallel ones. 
  }
  \label{fig-parallel}
\end{figure}

\smallskip 
\noindent
\textbf{Full results}. As expected, the time taken to compute answers decreases drastically (the behaviour is the same as for \qA). Perhaps more interestingly, we 
focus on how algorithms change when more parallel threads are allowed. In order to do this, we repeat the same reports made in the previous subsection, but this time for 
different levels of parallelism. More precisely, for each of our 11 queries and 4 different thread counts we report which of A*, BFS or DFS dominates in the time needed to compute the answers. 
As we see 
as more parallelism is allowed into the algorithms, both BFS and DFS start becoming more competitive compared to A*.
\begin{center}
  \begin{tabular}{l|lll}
    \hline
    Max parallel calls & A*  & BFS & DFS  \\ \hline
    1           &  11     & 3   & 4    \\
    10          &  7     & 3   & 3    \\
    20          &  7     & 3   & 3    \\
    40          &  6     & 4   & 5    \\
  \end{tabular}
\end{center}




\subsection{Returning paths}

So far we have only talked about finding pairs of nodes that form the answer of a property path query, as dictated by the SPARQL standard. 
However, our search algorithms can also be used to compute the entire \emph{path} between two nodes, and in fact we can get them 
at a very marginal cost: we already need to keep track of all the expansions, so we can produce paths simply by returning the IRIs corresponding to each of the 
requests made by our algorithm. 

Paths can be used as a justification for the answers, or to continue extracting more information afterwards. In the case of queries over Linked Data, 
we can even use the paths of queries to gather information about the structure of the Web itself.
For these reasons 
returning paths is a very sought-after functionality of graph query languages, and is present for example in the popular Neo4j engine \cite{neo4j}. 
Unfortunately, a language capable of returning (all) paths, or even (all) simple paths, is bound to be very complicated to evaluate \cite{ACP12,LM12}, and this is 
the reason why the SPARQL standard does not include such a functionality. 
In our case we have a natural workaround for this issue, as our search focuses on shortest paths, which are known to be easier to evaluate than simple paths. 


\begin{figure}
\centering
\scriptsize
\fontsize{0.25cm}{0.1cm}\selectfont{
\begin{verbatim}
dblpAuthor:Michael_Stonebraker
  ^dc:creator  dblpPub:conf/acm/MuthuswamyKZSPJ85
   dc:creator  dblpAuthor:Matthias_Jarke
   rdfs:label  "Matthias Jarke"
 
 
dblpAuthor:Michael_Stonebraker
  ^dc:creator  dblpPub:conf/dbvis/AikenCLSSW95
   dc:creator  dblpAuthor:Mybrid_Spalding
   rdfs:label  "Mybrid Spalding"


dblpAuthor:Michael_Stonebraker
  ^dc:creator  dblpPub:conf/vldb/StonebrakerABCCFLLMOORTZ05
   dc:creator  dblpAuthors:Adam_Batkin
   rdfs:label  "Adam Batkin"
\end{verbatim}
}
  \vspace{-8pt}
\caption{Paths for the 10th, 50th, and 200th answers found by A* on \qA.}
\label{fig:qAPaths}
\end{figure}

As an example of the usefulness of paths, 
it was by analysing paths that we inferred that A* 
normally produces better answers than DFS (because the paths are shorter).  
As an illustration, Figure \ref{fig:qAPaths}
presents paths witnessing the answers $10$, $50$, and $200$ of a run of the
query \qA with A*. From the query itself all that we can say is that
these three researchers are connected to M. Stonebraker by a coauthorship path
of arbitrary length. However, by looking at the paths we now know that they are
direct coauthors. On the other hand, the length of paths retrieved by DFS are
going to be much higher. For one run of \qA with DFS the lengths of 
the answers 10, 50 and 200 were respectively 14, 74, 312.

\section{Conclusions and future work}
\label{sec:Conclusions}

This paper presents the first fundamental study of the problem of computing property paths over the Web. We showed how to cast query answering as an 
AI search problem, and provided an optimal algorithm based on the classical A* algorithm. We provide strong theoretical and practical evidence that A* is a better 
alternative than both BFS and DFS in the context of Linked Data, and this can be sped up even further by allowing parallel execution threads. 

In terms of future work, we identify three main challenges we plan on tackling.

\smallskip
\noindent
\textbf{Using triple pattern fragments}. As noted in Section \ref{sec:Optimisations}, there are some issues with the Linked Data infrastructure; most notably, it does not provide all the information one would expect when dereferencing IRIs. While it is possible to alleviate this issue by using endpoints, since their uptime can be erratic, it was recently suggested that a more lightweight infrastructure of triple pattern fragments \cite{TPF} would be more appropriate for the task. In the future we plan to test how using triple pattern fragments affects the performance and accuracy of our algorithms when compared to the standard endpoint infrastructure.

\smallskip
\noindent
\textbf{Answering NautiLOD and LDQL queries with A*}. NautiLOD \cite{FPG15} is a traversal-based language proposed as 
an option to SPARLQ when querying Linked Data, in which one has more finer control on how is the Web going to be traversed. In the same spirit, LDQL \cite{HPe15} is another language aimed at controlling how data is to be retrieved, albeit much less powerful than NautiLOD. The interesting observation is that we can also cast the query evaluation problem for these languages as a search problem, and thus A* should also provide optimal query answering algorithms. In fact, the algorithm 
proposed in \cite{FPG15} is essentially what we define here as $k$-DFS, so one can naturally suspect that A* should provide a better behaviour.

\smallskip
\noindent
\textbf{A* in local computations}. Although we based our investigation in the context of Linked Data, there is some evidence that  our approach might have potential in the 
classical setting where data is available locally. The main reason is the fact that the currently available property path evaluation algorithms demand a lot of resources, especially when dealing with property paths that use the Kleene star operator, and current systems cannot easily cope with these requirements \cite{BaierDRV16}. On the other hand, we have seen that the memory usage of an A*-based algorithm is directly dependant on the amount of answers 
that need to be computed, and each answer requires an almost negligible amount of additional memory. This suggests that, in those cases when we do not need all the answers, an approach based on A* might be a better option.

\bigskip

\bibliographystyle{ACM-Reference-Format}
\bibliography{ref}

\end{document}